\documentclass{llncs}

\usepackage{amssymb}
\usepackage{amsmath}
\usepackage{textcomp}
\usepackage{array}
\usepackage{lineno}
\usepackage{color}
\usepackage{comment}

\usepackage{graphicx}

\newcommand{\ket}[1]{|#1\rangle}

\title{Unary probabilistic and quantum automata on promise problems\thanks{The arXiv number is 1502.01462.}}

\author{Aida Gainutdinova \and Abuzer Yakary{\i}lmaz\inst{}$^,$\thanks{Yakary{\i}lmaz was partially supported by CAPES with grant 88881.030338/2013-01.}}

\institute{Kazan Federal University, Kazan, Russia       
       \and
National Laboratory for Scientific Computing, Petr\'{o}polis, RJ, 25651-075, Brazil
\email{aida.ksu@gmail.com,abuzer@lncc.br}
}

\authorrunning{A. Gainutdinova \and A. Yakary{\i}lmaz} 

\begin{document}

\maketitle

\begin{abstract}
We continue the systematic investigation of probabilistic and quantum finite automata (PFAs and QFAs) on promise problems by focusing on unary languages. We show that bounded-error QFAs are more powerful than PFAs. But, in contrary to the binary problems, the computational powers of Las-Vegas QFAs and bounded-error PFAs are equivalent to deterministic finite automata (DFAs).  Lastly, we present a new family of unary promise problems with two parameters such that when fixing one parameter QFAs can be exponentially more succinct than PFAs and when fixing the other parameter PFAs can be exponentially more succinct than DFAs.
\end{abstract}

\section{Introduction}

Promise problems are generalizations of language recognition. The aim is, instead of separating one language from its complement, to separate any two disjoint languages. That is, the input is promised to be from the union of these two languages. Promise problems have served some important roles in the computational complexity. For example, it is not known whether the class BPP (BQP), bounded error probabilistic (quantum) polynomial time, has a complete problem, but, the class PromiseBPP (PromiseBQP), defined on promise problems, has some complete problems (see the surveys by Goldreich \cite{Gol06A} and Watrous \cite{Wat09A}). 

In automata theory, the promise problems has also appeared in many different forms. For example, in 1989  Condon and Lipton \cite{CL89} defined a promised version of emptiness problem for probabilistic finite automata (PFAs), and showed its undecidability by using a promised version of equality language ($\mathtt{EQ}=\{a^nb^n | n > 0\}$), solved by two-way bounded-error PFAs, which was also used to show that there is a weak constant-space interactive proof system for any recursive enumerable language. 

On the other hand, up to our knowledge, some systematic works on promise problems in automata theory have been started only recently. An initial result was given to compare exact quantum and deterministic pushdown automata \cite{MNYW05}, the former one was shown to be more powerful (see also \cite{Nak15} and \cite{NY15A} for the results in this direction). Then, the result given by Ambainis and Yakary{\i}lmaz \cite{AY12}, the state advantages of exact quantum finite automata (QFAs) over deterministic finite automata (DFAs) cannot be bounded in the case of unary promise problems, has stimulated the topic and a series of papers appeared on the succinctness of QFAs and other models \cite{ZQGLM13,GQZ14A,GQZ14B,ZGQ14A,BMP14B,AGKY14}. In parallel, the new results were given on classical and quantum automata models \cite{RY14A,GY14}:
\begin{itemize}
	\item There is a promise problem solved by exact two-way QFAs but not by any sublogarithmic probabilistic Turing machine (PTM).
	\item There is a promise problem solved by an exact two-way QFA in quadratic expected time, but not by any bounded-error $ o(\log \log n) $-space PTMs in polynomial expected time.
	\item There is a promise problem solvable by a Las Vegas realtime QFA, but not by any bounded-error  PFA.
	\item The computational power of deterministic, nondeterministic, alternating, and Las Vegas PFAs are the same and two-wayness does not help.
	\item On the contrary to tight quadratic gap in the case of language recognition, Las-Vegas PFAs can be exponentially more state efficient than DFAs.
	\item The state advantages of bound-error unary PFAs over DFAs cannot be bounded.
	\item There is a binary promise problem solved by bounded-error PFAs but not by any DFAs.
\end{itemize}

In this paper, we provide some new results regarding probabilistic and quantum automata on unary promise problems. We show that bounded-error QFAs are more powerful than PFAs. But, on contrary to the binary problems, the computational power of Las-Vegas QFAs and bounded-error PFAs are equivalent to DFAs. Lastly, we present a new family of unary promise problems with two parameters such that when fixing one parameter QFAs can be exponentially more succinct than PFAs and when fixing the other parameter PFAs can be exponentially more succinct than DFAs.

\section{Preliminaries}
\label{sec:preliminaries}

In this section, we provide the necessary background to follow the remaining part. Firstly we give the definitions of models and the notion of promise problems. Then, we give the basics of Markov chain theory which will be used in some proofs. 

\subsection{Definitions}

A PFA $ \mathcal{P} $ is a 5-tuple
$
	\mathcal{P} = (Q,\Sigma,\{ A_\sigma \mid \sigma \in \Sigma \},v_0,Q_a),
$ where
\begin{itemize}
	\item $  Q $ is the set of states,
	\item $ \Sigma $ is the input alphabet,
	\item $ v_0 $ is a $ |Q| $-dimensional stochastic initial column vector that represents the initial probability distribution of the states at the beginning of the computation,
	\item $ A_\sigma $ is a (left) stochastic transition matrix for symbol $ \sigma \in \Sigma $ where $ A_\sigma(j,i) $ represents the probability of going from the $i$th state to the $j$th state after reading $ \sigma $, and
	\item $ Q_a $ is the set of the accepting states.
\end{itemize}
The computation of $ \mathcal{P} $ on the input $ w \in \Sigma^* $ can be traced by a stochastic column vector, i.e.
\[
	v_j = A_{w_j} v_{j-1},
\]
where $ 1 \leq j \leq |w| $. After reading the whole input, the final state is $ v_{|w|} $. Based on this, we can calculate the accepting probability of $ w $ by $ \mathcal{P} $, denoted $ f_{\mathcal{P}}(w) $, as follows:
\[
	 f_{\mathcal{P}}(w) = \sum_{q_j \in Q_a} v_{|w|}(j).
\]
If all stochastic elements of a PFA are restricted to have only 0s and 1s, then we obtain DFA that starts in a certain state and switches to only one state in each step, and so the computation ends in only a single state. An input is accepted by a DFA if the final state is an accepting state.

There are different kinds of quantum finite automata (QFAs) models in the literature. The general ones (e.g. \cite{Hir10,YS11A,AG05}) can exactly simulate PFAs (see \cite{SayY14} for a pedagogical proof). In this paper, we present our results based on the known simplest QFA model, called Moore-Crutcfield QFA \cite{MC00}. Therefore, we only provide its definition. We assume the reader knows the basics of quantum computation (see \cite{SayY14} for a quick review and \cite{NC00} for a complete reference).

A MCQFA $\mathcal{M}$ is 5-tuple
$
	\mathcal{M} = (Q,\Sigma,\{ U_\sigma \mid \sigma \in \Sigma \},\ket{v_0},Q_a)
$
where, different from a PFA,
\begin{itemize}
	\item $ \ket{v_0} $ is a norm-1 complex-valued column initial vector that can be a superposition of states and represents the initial quantum state of $ \mathcal{M} $ at the beginning of the computation, and,
	\item $ U_\sigma $ is a unitary transition matrix for symbol $ \sigma \in \Sigma $ where $ U_\sigma(j,i) $ represents the amplitude of going from the $i$th state to the $j$th state after reading $ \sigma $. 
\end{itemize}
Traditionally, vectors are represented with ``ket'' notation ($\ket{\cdot}$) in quantum mechanics and computations. The computation of $ \mathcal{M} $ on the input $ w \in \Sigma^* $ can be traced by a norm-1 complex-valued column vector, i.e.
\[
	\ket{v_j} = U_{w_j} \ket{v_{j-1}},
\]
where $ 1 \leq j \leq |w| $. After reading the whole input, the final quantum state is $\ket{v_{|w|}} $. Based on this, a measurement operator is applied to see whether the automaton in an accepting or non-accepting state.  The accepting probability of $ w $ by $ \mathcal{M} $ is calculated as:
\[
	 f_{\mathcal{M}}(w) = \sum_{q_j \in Q_a} | \ket{v_{|w|}}(j) |^2.
\]

A Las Vegas PFA (or QFA) never gives a wrong decision, instead giving the decision of ``don't know''. Formally, its set of states is divided into three disjoint sets, the set of accepting states ($ Q_a $), the set of neutral states ($ Q_n $), and the set of rejecting states ($ Q_r = Q \setminus Q_a \cup Q_n $). At the end of the computation, the decision of ``don't know'' is given if the automaton ends with an neutral state. The probability of giving the decision of ``don't know''  (rejection) is calculated similar to the accepting probability by using $ Q_n $ ($Q_r$) instead of $Q_a$. 

A promise problem $ \mathtt{P} \subset \Sigma^* $ is composed by two disjoint languages $ \mathtt{P_{yes}} $ and $ \mathtt{P_{no}} $, where the former one is called a set of yes-instances and the latter one is called a set of no-instances. 

A promise problem is said to be solved by a DFA if any yes-instance is accepted and any no-instance is rejected. A promise problem is said to be solved by a PFA or QFA with error bound $ \epsilon < \frac{1}{2} $ if any yes-instance is accepted with probability at least $ 1 - \epsilon $ and any no-instance is rejected with probability at least $ 1- \epsilon $. If all yes-instances are accepted exactly, then it is said the promise problem is solved with one-sided bounded error. In this case, the error bound can be greater than $ \frac{1}{2} $ but it must be less than 1, i.e. $ \epsilon <1 $.  Lastly, a promise problem is said to be solved by a Las Vegas PFA or QFA with success probability $ p>0 $,
\begin{itemize}
	\item if any yes-instance is accepted with probability at least $ p $ and it is rejected with probability 0, and,
	\item if any no-instance is rejected with probability at least $ p $ and it is accepted with probability 0.
\end{itemize}
In the case of promise problems, we do not care about the decisions on the strings from $ \Sigma^* \setminus \mathtt{P} $.


\subsection{The theory of Markov Chains}
\label{Markov_chain} 

The computation of a unary PFA can be described by a Markov chain. Here we present some basic facts and results from theory of Markov chains that will be used in some proofs. We refer the reader to \cite{ks60} for more details and \cite{AF98} and \cite{MPP01B} for some similar applications.
 
The states of a Markov chain are divided into ergodic and transient states. An {\em ergodic set of states} is a set which a process cannot leave once it has entered, a {\em transient set of states} is a set which a process can leave, but cannot return once it has left. An arbitrary Markov chain has at least one ergodic set.  If a Markov chain $C$ has more than one ergodic set, then there is absolutely no interaction between these sets. Hence we have two or more unrelated Markov chains lumped together and can be studied separately. If a Markov chain consists of a single ergodic set, then the chain is called an {\em ergodic chain}. According to the classification mentioned above, every ergodic chain is either regular or cyclic (see below).

If an ergodic chain is regular, then for sufficiently high powers of the state transition matrix, $M$ has only positive elements. Thus, no matter where the process starts, after a sufficiently large number of steps it can be in any state. Moreover, there is a limiting vector of probabilities of being in the states of the chain, that does not depend on the initial state. 

If a Markov chain is cyclic, then the chain has a period $t$ and all its states are subdivided into $t$ cyclic subsets $(t>1)$. For a given starting state a process moves through the cyclic subsets in a definite order, returning to the subset with the starting state after every $t$ steps. It is known that after sufficient time has elapsed, the process can be in any state of the cyclic subset appropriate for the moment. Hence, for each of  $t$ cyclic subsets the $t$-th power of the state transition matrix $M^t$ describes a regular Markov chain. Moreover, if an ergodic chain is a cyclic chain with the period $t$, it has at least $t$ states.

Let $C_1, \dots, C_l$ be cyclic subsets of states of Markov chain with periods $t_1, \dots, t_l$, respectively, and $D$ be the least common multiple of $t_1, \dots, t_l$. 
For each cyclic subset $C$ after every $ D $ steps, the process can be  in any state of $C$ and the $D$th power of $M$ describes a regular  Markov chain for this subset.
From the theory of Markov chains it is known  that there exists an $\alpha_{acc}$ such that $\lim_{r \to \infty}\alpha^{r \cdot D}_{acc}=\alpha_{acc}$, where $\alpha_{acc}^i$ represents the probability of process being in accepting state(s) after the $i$th step. Hence, for any $\delta>0$, there exists an $ r_0>0 $ such that  \[|\alpha^{r \cdot D}_{acc}-\alpha^{r' \cdot  D}_{acc}|<\delta\]for any $r, r' > r_0$.

Moreover, since $ \alpha_{acc}^{r \cdot D} $ has a limit point $ \alpha_{acc} $, each $ \alpha_{acc}^{r \cdot D+j} $ has also a limit point, say $ \alpha_{acc(j)} $ for any $ j \in \{ 1,\ldots,D-1 \} $. 


\section{The computational power of unary PFAs and QFAs}

First we show that any unary promise problem solved by a QFA exactly (without error) can also be solved by DFAs. 

\begin{theorem}
	If a unary promise problem $ \sf P = (P_{yes},P_{no}) $ is solved by a QFA exactly, then it is also solved by a DFA.
\end{theorem}
\begin{proof}
	Let $ \cal M $ be the a QFA solving $\sf P$ exactly. The automaton $ \cal M $ also defines a language with cutpoint 0, say $\tt  L $, i.e. any string accepted with a non-zero (zero) probability is a member (non-member). Then, we can easily obtain the following two facts:
\begin{itemize}
	\item Since each yes-instance of  $\tt P$ is accepted with probability 1, it is also a member of $ \tt L $. Thus, $ \tt P_{yes} $ is a subset of $ \tt L $.
	\item Since each no-instance of $ \tt P $ is accepted with probability 0, it is also a member of $ \tt \overline{L} $. Thus, $ \tt P_{no} $ is a subset of $  \tt \overline{L} $.
\end{itemize}	
	Any unary language defined by a QFA with cutpoint 0 (or equivalently recognized by a nondeterministic QFA \cite{YS10A}) is a unary exclusive language and it is known that any such language is regular (Page 89 of \cite{SS78}). Thus, $ \tt L $ is a unary regular language and there is a DFA, say $\cal D$, recognizing $ \tt L $. So, $ \cal D $ can also solve promise problem $ \tt P $: $ \cal D $ accepts all members of $ \tt L $ including all $ \tt P_{yes} $ and it rejects all members of $ \tt \overline{L} $ including all $ \tt P_{no} $.
\qed\end{proof}

We can extend this result also for Las Vegas QFAs. 

\begin{theorem}
	If a unary promise problem $ \tt P = (P_{yes},P_{no}) $ is solvable by a Las Vegas QFA  with a success probability $ p>0 $, then it is also solvable by a DFA.
\end{theorem}
\begin{proof}
	Let $ \cal M $ be our Las Vegas QFA solving $\tt P$ with success probability $ p>0 $. We can obtain a new QFA $ \cal M' $ by modifying $ \cal M $ as follows:  $ \cal M' $ rejects the input when entering a neutral state at the end of the computation. Then, any member of $ \tt P_{yes} $ is accepted by $ \cal M$ with probability at least $ p $ and any member of $ \tt P_{no} $ is accepted by $ \cal M' $ with probability 0. After this, we can consider $ \cal M'$ as a  nondeterministic QFA and follow the same reasoning given in the previous proof.
\qed\end{proof}

Since Las Vegas QFAs and DFAs define the same class of unary promise problems, one may ask how much state efficient QFAs can be over DFAs. Due to the result of Ambainis and Yakary{\i}lmaz \cite{AY12}, we know that the gap (on unary promise problems) cannot be bounded. (Note that, in the case of language recognition, there is no gap between exact QFA and DFA \cite{Kla00} and the gap can be at most exponential between bounded-error QFAs and DFAs (see e.g. \cite{AG00}).) On the other hand, as mentioned before, over binary promise problems, Las Vegas QFAs are known to be more powerful than bounded-error PFAs \cite{RY14A}.  An open question here is whether exact QFAs can solve a binary promise problem that is beyond the capabilities of DFAs.

Las Vegas PFAs and DFAs have the same computational power even on binary promise problems and the tight gap on the number of states is exponential \cite{GY14}. Currently we do not know whether this bound can be improved on unary case and we leave it as a future work. Here we show that making two-sided errors does not help to solve a unary promise problem that is beyond of the capability of DFAs. However, remark that, the state efficiency of bounded-error unary PFAs over unary DFAs also cannot be bounded \cite{GY14}. 

\begin{theorem}
	If a unary promise problem $ \tt P = (P_{yes},P_{no}) $ is solved by a PFA, say $\mathcal{P}$, with error bound $ \epsilon < \frac{1}{2} $, then it is also solvable by a DFA.
\end{theorem}
\begin{proof}
	The computation of $ \mathcal{P} $ can be modelled as a Markov chain. Let $  \mathcal{P} $ has $ n $ states and $ D $ be the least common multiple of periods of cycles of Markov chain (see the Section \ref{Markov_chain}). So, $ \mathcal{P} $ has $ D $ limiting accepting probabilities as described in Section \ref{Markov_chain}, say
	\[
		\alpha_{acc(0)}, \alpha_{acc(1)}, \ldots, \alpha_{acc(D-1)}.
	\]
	For any small $ \delta >0 $, there is an integer $ r_0 $ such that,  for each $ j \in \{0,\ldots,D-1\} $, we have the inequality $ | f_{\mathcal{M}}(a^{r\cdot D + j}) - \alpha_{acc(j)} | < \delta $ for all $ r \geq r_0 $. Let's pick a $ \delta' >0 $ such that, for any index $ i \in \{0,\ldots,D-1\} $, the interval $ | \alpha_{acc(i)} -\delta'  | $ does contain at most one of the points $ 1 - \epsilon $ and $ \epsilon $, which is always possible since the gap between these two points ($1-2\epsilon$) is non-zero. For this $ \delta' $, we also have a $ r'_0 $, such that, for any  $ j \in \{0,\ldots,D-1\} $, $ f_{\mathcal{M}}(a^{r\cdot D + j}) $ is in the interval $ | \alpha_{acc(j)} -\delta'  | $ for all $ r \geq r'_0 $. 
	
	We can classify $ \alpha_{acc}(j) $ as follows:
	\begin{itemize}
		\item It is at least $ \frac{1}{2} $. Then,  $ f_{\cal M} (a^{r\cdot D + j}) $ cannot be $ \epsilon $ or less than $ \epsilon $ for any $r \geq r_0$. 
		\item It is less than $ \frac{1}{2} $. Then, $ f_{\cal M} (a^{r\cdot D + j}) $ cannot be $ 1- \epsilon $ or greater than $ 1- \epsilon $ for any $ r \geq r_0 $. 
	\end{itemize}
	Thus, a $D $-state cyclic DFA with the following state transitions
	\[
		q_0 \rightarrow q_1 \rightarrow \cdots \rightarrow q_j \rightarrow \cdots \rightarrow q_{D-1} \rightarrow q_0
	\] can easily follow the periodicity of $  \cal P $. Moreover, if $ \alpha_{acc}(j) $ belongs the first (second) class of the above, then $ q_j $ is an accepting (a non-accepting) state. Thus, our cyclic DFA can 	give the same decisions of $ \cal P $ on the promised strings with length at least $ r_0 \cdot D $. The remaining (and shorter) promised strings form a finite set and a DFA with $ (r_0' \cdot D -1) $ states can give appropriate decisions on them. Therefore, by combining two DFAs, we can get a DFA with $ r_0' \cdot D +D $ states that solves the promise problem $ \tt P $.
\qed\end{proof}

Now we show that unary QFAs can define more promise problems than PFAs when the machines can err. We present our quantum result by a 2-state MCQFA. Then, we give our impossibility result for unary PFAs. 

Let $ \varphi $ be a rotation angle which is an irrational fraction of $2 \pi$.  For any $\theta \in (0,\frac{\pi}{4})$, we define a unary promise problem $\tt L^{\theta}=\{L^{\theta}_{yes},L^{\theta}_{no}\}$ as
\begin{itemize}
	\item $ \mathtt{L^{\theta}_{yes}}=\{a^k \mid k\varphi \in [l\pi -\theta, l\pi +\theta] \mbox{ for some } l \geq 0\},$
	\item $ \mathtt{L^{\theta}_{no}}=\{a^k \mid k\varphi \in  [l\pi +\frac{\pi}{2} -\theta, l\pi  + \frac{\pi}{2} +\theta] \mbox{ for some } l \geq 0\}.$
\end{itemize}

\begin{theorem} 
	\label{thm:2-state-QFA}
	There is a 2-state MCQFA $\cal M$ solving the promise problem $\tt L^{\theta}$ with error bound $ \sin^2 \theta < \frac{1}{2} $. Moreover, $ \cal M $ is defined only with real number transitions.
\end{theorem}
\begin{proof}
	Let $ \{q_1,q_2\} $ be the set of states of $ \cal M $ and $ q_1 $ be the initial and the only accepting state. The unitary operation is a rotation on $ \ket{q_1}-\ket{q_2} $ plane with the angle $\varphi$. (Note that, there are infinitely many $\varphi$ whose rotation matrices contain only rational numbers, e.g. $ \arcsin \frac{3}{5}, \arcsin \frac{5}{13}, \arcsin \frac{7}{25} $, etc.). It is straightforward that, after reading $a^k$, the final quantum state becomes
	\[
		\ket{v_k} = \cos(k\varphi) \ket{q_1} + 	\sin(k\varphi) \ket{q_2},
	\] 
	and so $a^k$ is accepted by $ M $ with probability  $\cos^2 ( k\varphi ) $. It is clear that $ \mathcal{M} $ takes $ a^k $ and leaves it as $ \ket{v_k} $ before the measurement, which can be seen as a map from an angle to a point on the unit circle. Therefore, the bounds on $ k\varphi $ give similar bounds on $ \ket{v_k} $, that allows $ \cal M $ to solve the problem with bounded error. Now, we show that $\sin^2(\theta) < \frac{1}{2} $ can be a bound on the error.

	If $ a^k $ is a yes-instance, we have $ \cos \theta \leq | \cos (k\varphi) | \leq 1 $. Then, the accepting probability can be bounded as $ \cos^2 \theta \leq \cos^2 (k\varphi) \leq 1 $. That is, any yes-instance is accepted with probability at least $ \cos^2(\theta) $, which is equal to $ 1 - \sin^2 (\theta) $. In other word, the error for yes-instances can be at most $ \sin^2 \theta $.
	
	If $ a^k $ is a no-instance, $ 0 \leq | \cos (k\varphi) | \leq \sin \theta $. Then, the accepting probability can be bounded as $ 0 \leq \cos^2 (k\varphi) \leq  \sin^2 \theta  $. That is, any no-instance is accepted with probability at most $ \sin^2 \theta $, i.e. the error can be at most $ \sin^2 \theta  $ for any no-instance.
\qed\end{proof}

\begin{theorem} 
	\label{thm:no-pfa}
	There exists no PFA solving the promise problem $\tt L^{\theta}$ for any error bound $ \epsilon < \frac{1}{2} $.
\end{theorem}
\begin{proof}
	Let us prove by contradiction. Let $\mathcal{P}=(Q,\{a\},M,v_0,Q_a)$ be a PFA solving $\tt L^{\theta}$ with the error bound $ \epsilon < \frac{1}{2} $.
Since $ \tt L^{\theta}$ is a unary problem, so the computation of $\cal P$ can be described by a Markov chain and the states of $\cal P$ can be classified as described in Section \ref{Markov_chain}. Let $C_1, \dots, C_l$ be cyclic subsets of states of Markov chain with periods $t_1, \dots, t_l$, respectively, and $D$ be the least common multiple of $t_1, \dots, t_l$.

We pick a yes-instance $a^n \in\tt  L^{\theta}_{yes}$ and define the set $A^n =  \{a^{n+kD} \mid k \in \mathbb{Z^+}\}$. Now, we show that $ A^n $ contains some no-instances, i.e. $A^n\cap L^{\theta}_{no}\neq\emptyset $. \\

\noindent
\textit{Claim.} $A^n\cap L^{\theta}_{no}\neq\emptyset. $ \\

\noindent
{\em Proof of the claim.} As verified from the definition of $\tt L^{\theta}$, each string can be associated to a point on the unit circle. Let $\gamma_n$ be the angle of this point corresponding to our yes-instance $a^n$. So we have that $\gamma_n \in [-\theta,\theta]\cup [\pi-\theta,\pi+\theta]$. From now on, we consider all angles up to $2\pi$ and will omit the period $2\pi$ from the value of angles.  An input $a^j$ is a no-instance ($a^j \in L^{\theta}_{no}$) if and only if $\gamma_{j}\in [\frac{\pi}{2}-\theta, \frac{\pi}{2}+\theta]\cup[3 \frac{\pi}{2}-\theta, 3\frac{\pi}{2}+\theta] $.  We need to show that there is an $l\in N$ such that $a^{n+lD}\in L^{\theta}_{no}$, that means $\gamma_{n+lD} \in  [\frac{\pi}{2}-\theta, \frac{\pi}{2}+\theta]\cup[3 \frac{\pi}{2}-\theta, 3\frac{\pi}{2}+\theta]$. 
 
Reading $D$ letters of the input  corresponds to  a rotation on the circle by the angle $D\varphi$. Let $\beta = D\varphi-2\pi m$ for some $m\in \mathbb{N} $ satisfying $\beta\in (0, 2\pi)$. Since $\varphi$ is an irrational multiple of $\pi$, $\beta$ is also an irrational multiple of $\pi$. It is a well known fact that a rotation with an angle of irrational multiple of $ \pi $ is dense on the unit circle. So the points corresponding to $ \{ a^{Dk} \mid k \in \mathbb{Z}^+ \} $ are dense on the unit circle (and none of two strings from this set corresponds to the same point on the unit circle).

So for each point $\gamma_n \in [-\theta, \theta]$ (or for each point $ \gamma_n \in [\pi-\theta,\pi+\theta]$), there exists an $l \in \mathbb{Z}^+$ such that $\gamma_{n+lD} \in  [\frac{\pi}{2}-\theta, \frac{\pi}{2}+\theta]\cup[3 \frac{\pi}{2}-\theta, 3\frac{\pi}{2}+\theta]$. Therefore, the set $ A^n =  \{a^{n+kD} \mid k \in \mathbb{Z^+}\} $ contains some no-instances. This completes the proof of the claim. $\lhd$
 
The final state of $ a^n $ is $ v_n = M^n v_0 $. Since there is no assumption on the length of $ a^n $, it can be arbitrarily long. Assume that $ n $ is sufficiently big providing that
\[
	| f_P(a^{n+rD}) - f_P(a^{n+r'D}) | < \frac{1}{2} - \epsilon	
\]
for any $ r,r' $. Remember from Section \ref{Markov_chain} that this assumption follows from Markov chain theory and the bound approaches to 0 when $ n \rightarrow \infty $. If a promise problem is solvable with an error bound $ \epsilon $, then the difference between  the accepting probabilities of a yes-instance and no-instance can be at least $ 1 - 2\epsilon $. The set $ A^n $ has at least one no-instance whose accepting probability cannot be less than $ \frac{1}{2} $, since (i) the minimal accepting probability for a member is $ 1-\epsilon $ and (ii) we can obtain at least $ \frac{1}{2} $ if we go away from $ 1-\epsilon $ with $ \frac{1}{2} - \epsilon	 $. However, this no-instance must be accepted with a probability at most $ \epsilon < \frac{1}{2} $. Therefore, the PFA $ \cal P $ cannot solve the promise problem $\tt L^{\theta}$ with an error bound $ \epsilon < \frac{1}{2} $.
\qed\end{proof}


\section{Succinctness}

For each $ n \in \mathbb{Z^+} $,  we define a family of unary promise problems $F_n = \{ \mathtt{L^{k,n}} \mid k \in \mathbb{Z^+}\} $ as follows. Let $p_j$ be the $j$-th prime, $P_{k,n}=\{p_n, p_{n+1}, \dots, p_{n+k-1}\} $ be the set of primes from  $n$-th to $(n+k-1)$-th  one, and $N=p_n\cdot p_{n+1}\cdots p_{n+k-1} $.

The promise problem $\tt L^{k,n}=\{L^{k,n}_{yes}, L^{k,n}_{no}\} $ is defined as
\begin{itemize}
	\item $ \mathtt{L^{k,n}_{yes}}=\{a^m \mid  m  \equiv 0 \mod N$ \}  and
	\item $ \mathtt{L^{k,n}_{no}}=\{a^m \mid m \mod p_j \in \left[\frac{p_j}{8}, \frac{3p_j}{8}\right] \cup \left[\frac{5p_j}{8},\frac{7p_j}{8}\right] $ for at least $\frac{2k}{3}$ different $p_j$ from the set $P_{k,n}$\}.
\end{itemize}
Here we can use \textit{Chinese remainder theorem} to show that the number of no-instances is also infinitely many.

\begin{lemma}	
	There are infinitely many strings in $\tt{L^{k,n}_{no}}$. 
\end{lemma}
\begin{proof}
	If positive integers $p_1, p_2, \dots, p_n$  are pairwise coprime, then for any integers $r_1, r_2, \dots, r_n$ satisfying $ 0 \leq r_i < p_i$ ($ i \in \{1, 2, \dots, n\}$), there exists a number $K$, such that $ K =r_i \pmod {p_i}$ for each $i \in \{1, 2, \dots, n\}$. Moreover, any such $K$ is congruent  modulo the product, $N = p_1 \cdots p_n$. That is all numbers of the form $K+N\cdot m$ will satisfy this condition, where $ m \in \mathbb{Z}^+ $.
\qed\end{proof}

\begin{theorem}
	For any $n \in \mathbb{Z^+}$, the promise problem $\tt L^{k,n}$ can be solvable by a  $2k$-state MCQFA, say $ \mathcal{M}_{k,n} $, such that yes-instances are accepted exactly and no-instance are rejected with probability at least $ \frac{1}{3} $.
\end{theorem}
\begin{proof}
	We use the technique given in \cite{AF98,AGKMP05}. The set of states of automaton $\mathcal{M}_{k,n}$ is $ \{ q_1^0,q_1^1, \dots, q_k^0, q_k^1 \} $ and the ones with superscript ``0'' are the accepting states. The initial quantum state is \[
	\ket{v_0} = \frac{1}{\sqrt{k}} \ket{q_1^0} + \frac{1}{\sqrt{k}} \ket{q_2^0} + \cdots + \frac{1}{\sqrt{k}} \ket{q_k^0}.
\] 
During reading the input, the states $ \ket{q_j^0} $ and $ \ket{q_j^1} $ form a small MCQFA  isolated from the others, where $ 1 \leq j \leq k $. For each letter $a$, a rotation with angle $ \frac{2\pi}{p_j} $ is applied on $ \{ \ket{q_j^0},\ket{q_j^1} \} $:
\[
	U_j=\left( \begin{array}{cc}
\cos(2\pi/p_j) & \sin(2\pi/p_j)\\
-\sin(2\pi/p_j) & \cos(2\pi/p_j)
\end{array}\right).
\]
Then, the overall transition matrix is
\[
	U=\left( \begin{array}{cccc}
U_1     & {\bf 0} & \cdots &  {\bf 0}\\
{\bf 0} & U_{2}     & \cdots &  {\bf 0}\\
\vdots  & \vdots  & \ddots &  \vdots \\
{\bf 0} & {\bf 0} & \cdots &  U_{k}
\end{array}\right)
\]
where ${\bf 0}$ denotes $2 \times 2$ zero matrix. 

For any input $a^m$ the final state of $\mathcal{M}_{k,n}$ is 
\[
	\ket{v_m} = \frac{1}{\sqrt{k}} \sum\limits_{j=1}^k\bigl( \cos\bigl( m\frac{2\pi}{p_j}\bigr) \ket{q_j^0} + \sin( m\frac{2\pi}{p_j} ) \ket{q_j^1}\bigr).
\]

For any yes-instance, $ m $ is multiple of $ N $ and so each $ m\frac{2\pi}{p_j} $ will be a multiple of $ 2\pi $. Then, the final state is in a superposition of only the accepting states, i.e.
\[
	\ket{v_m} = \frac{1}{\sqrt{k}} \sum\limits_{j=1}^k \ket{q_j^0},
\]
and so the input is accepted with probability 1. 
 
For any no-instance, on the other hand, it holds that $ (m \mod p_j) $ is in $  \left[\frac{p_j}{8}, \frac{3p_j}{8}\right] \cup \left[\frac{5p_j}{8},\frac{7p_j}{8}\right] $ for at least $ \frac{2k}{3} $ different $p_j$'s from the set $P_{k,n}$. If $ p_j $ is one of them, then its contribution to the overall rejecting probability is given by 
 \[
 	 \frac{1}{\sqrt{k}} \sin^2\bigl( m\frac{2\pi}{p_j}\bigr)
 \]
which takes its minimum value $ \frac{1}{2k} $ when  $  (m \mod p_j) $ is equal to one of the border. Since there are at least $ \frac{2k}{3} $ of them, the overall rejecting probability is at least $ \frac{1}{3} $.
\qed\end{proof}

\begin{theorem}
	\label{thm:state-bound-pfa}
	Any bounded-error PFA solving the promise problem $\tt L^{k,n}$ needs $\Omega(k(n+k)\log n)$ states. (See Appendix \ref{app:state-bound-pfa})
\end{theorem}

\begin{theorem} For any $n>0$, there is a $O(k(n+k)\log (n+k))$-state PFA $\mathcal{P}_{k,n}$ solving the promise problem   $ \mathtt{L^{k,n}} $ with one-sided error bound $ \frac{1}{3} $.
\end{theorem}
\begin{proof}
Let $ \mathcal{P}_{k,n} $, shortly $ \mathcal{P} $, be $ (Q,\{a\},\{ A_a \},v_0,Q_a) $, where
\begin{itemize}
	\item  $Q=\{q_{i,j} \mid i=1, \dots, k, j=0, \dots, p_{n+i-1}-1\}$ and $p_{n}, \dots, p_{n+k-1}$ are the primes from  the set $P_{k,n}$, 
	\item $ v_0 $ is the initial probabilistic state such that the automaton is in the state  $q_{i,0}$ with the probability $\frac{1}{k}$ for each $i=1, \dots, k$, and,
	\item $ Q_a = \{ q_{i,0} \mid i=1, \dots, k \} $.
\end{itemize}
The transitions of $ \mathcal{P} $ are deterministic: after reading each letter, it switches from state $ q_{i,j} $ to $ q_{i,j+1\pmod {p_{n+i-1}}} $. In fact, $\mathcal{P}$ executes $ k $ copies of DFAs with equal probability. The aim of the $ i $-th DFA is to determine whether the length of the input is equivalent to zero in mod $ p_{n+i-1} $. By construction it is clear that $\cal P$ accepts any yes-instance with the probability 1 and any no-instance with probability at most $ \frac{1}{3} $. 

The number of states is $ |Q| = p_n+\cdots +p_{n+k-1}$. It is known \cite{AT76} that  the $n$-th  prime number $p_n$ satisfies $p_n = \Theta( n\log(n))$ and so
\[
	|Q| = \sum_{x=n}^{n+k-1} p_x\leq O(k(n+k)\log(n+k)).
\]
\qed\end{proof}

Now, we give a lower and upper bound for DFAs.

\begin{theorem} 
	\label{thm:dfa-lowerbound}
	For any $n>0$, any DFA solving the promise problem $ \tt L^{k,n} $ needs $\Omega(n\log(n))^{\frac{k}{3}}$ states. (See Appendix \ref{app:dfa-lowerbound})
\end{theorem}

\begin{theorem} 
	\label{thm:dfa-upperbound}
	For any $n>0$, there is a $O((n+k/3)\log (n+k/3))^{k/3}$-state DFA $\mathcal{D}_{k,n}$ solving the promise problem   $L^{k,n}\in F_n$. (See Appendix \ref{app:dfa-upperbound})
\end{theorem}


\begin{figure}[!ht]
	\centering
		\[
			\begin{array}{|c|c|c|c|}
				\hline
				& \mbox{DFA} & \mbox{PFA} & \mbox{QFA}
				\\
				\hline
				\mbox{lower bounds} & 
				\Omega( n \log n )^{\frac{k}{3}} &
				\Omega(k(n+k)\log n)
				&
				1
				\\
				\hline
				\mbox{upper bounds} &
				O((n+\frac{k}{3})\log (n+\frac{k}{3}))^{\frac{k}{3}}
				&
				O(k(n+k)\log(n+k))	&
				2k		
				\\
				\hline
			\end{array}
		\]
	\caption{The summary of upper and lower state bounds for $ \tt L^{k,n} $}
	\label{fig:summary}
\end{figure}

We give the summary of the results in Figure \ref{fig:summary}. The bounds for DFAs and PFAs are almost tight and currently we do not know any better bound for QFAs. Moreover, if we pick $ n = 2^k $, then we obtain an exponential gap between QFAs and PFAs. On the other hand, if we pick $ n=k $, then we obtain an exponential gap between PFAs and DFAs.

\bibliographystyle{plain}
\bibliography{tcs}


\newpage

\appendix

\section{The proof of Theorem \ref{thm:state-bound-pfa}}
\label{app:state-bound-pfa}

Assume that  $ \mathcal{P}^{k,n} $, shortly $ \mathcal{P} $, is a PFA solving the promise problem $\tt L^{k,n}$ with error bound $ \epsilon < \frac{1}{2} $.  	
	Since $ \mathcal{P} $ is a unary automaton, its computation can be described by a Markov chain. Following the classification of states of Markov chain described in Section \ref{Markov_chain}, we know that, in the set of states of $\mathcal{P}$, there must exist a (some) cyclic subset(s) $C_1, \dots, C_l$  of states with periods $t_1, \dots, t_l$, respectively. Let $D$ be the least common multiple of $t_1, \dots, t_l$. Let fix an integer $r>r_0$ as a multiple of $N$ for sufficiently big $r_0$ and consider the sequence of stochastic vectors $\tilde{v}=(v_{r}, v_{r+1}, \dots)$, where $v_i$ is the state vector representing the probability distribution over the states in the $i$-th step.  Since $r$ is a multiple of $N$, $a^r \in \tt L^{k,n}_{yes}$.  The sequence $\tilde{v}$ can be divided into $D$ sub-sequences $\tilde{v}^0, \tilde{v}^1,  \dots , \tilde{v}^{D-1},  $ where $\tilde{v}^i=(v_{r+i}, v_{r+i+D}, v_{r+i+2D}, \dots)$. During the process,  the state vector moves cyclically through sub-sequences as
	\[
	\tilde{v}^0 \rightarrow \tilde{v}^1 \rightarrow \cdots \rightarrow \tilde{v}^{D-1} \rightarrow \tilde{v}^0 \rightarrow \cdots, 
\]
returning to the same subsequence after each $D$ steps. Moreover, for each $i=0, \dots, D-1$ there exists a limiting vector $u_{i}$ such that  the sequence $\tilde{v}^i$ converges to $u_{i}$.

Let $g=\gcd(N,D)$ and $D'=\frac{D}{g}.$ Since $N$ is a multiple of $g$, $g$ can be either 1 or a product of some primes from the set $P_{k,n}$. Let $S'=\{i\cdot g \mid i=0,\dots,D'-1\}$ and $S^{accept}=\{{i\cdot N \pmod D} \mid i\geq 0\}$, the set composed by the indices of the sub-sequences that include the state vectors $v_{i N}$.

We will show that $S'=S^{accept}$. Since $N$ is multiple of $g$, $ S^{accept}\subseteq S' $. Because  $|S'|=D'$, we only need to show $|S^{accept}|\geq D'$ to obtain $S'=S^{accept}$.

Firstly we show that if $i$ satisfies $i\cdot N \pmod D=0$, then $i$ must be a multiple of $D'$. 
If $i \cdot N \pmod D \equiv 0$, then $i N= j D$ for some $ j \geq 0 $. Since $D=D'g$, 
$j = \frac{i N}{D' g}$. We know that $\gcd(N,D')=1$. Hence $i$ is a multiple $D'.$ 

For two different $i_1,i_2 $, i.e. $ 0\leq i_1<i_2 <D'$, we must have $i_1 N\neq i_2 N \pmod D$. Otherwise we can have $(i_1-i_2)N \equiv 0\pmod D$ and so $i_1-i_2 $ must be a multiple of $D'$. But, this is a contradiction.

Since we have a different value of  $(i \cdot N\bmod D)$ for each $i \in \{0 \dots, D'-1\}$, $|S^{accept}|\geq D'$ and so  $S'=S^{accept}$. Then, we can follow that for each $i\in S'$, we have $\sum_{q_j \in Q_a} u_i(j)\geq 1-\epsilon$.

~\\
\noindent
\textit{Claim.} The number $g$ is at least $p_n\cdot p_{n+1} \cdots p_{n+\left\lceil  \frac{k}{3}\right\rceil} $.

~\\
\noindent
{\em Proof of the claim.}
Suppose that $g<p_n\cdot p_{n+1} \cdots p_{n+\left\lceil  \frac{k}{3}\right\rceil}. $ 
Then there are at least  $\frac{2k}{3}$ primes from $P_{k,n}$ which are not  multiples of $g$.  Let the set $R\subseteq P_{k,n}$ contains the primes not dividing $ g $, i.e. 
\[
	R=\{p_{i_j} \mid p_{i_j}\in P_{k,n}, g \not\equiv 0 \pmod {p_{i_j}} , j=1\dots,s\} ~~ \left(s\geq\frac{2k}{3}\right).
\]
 We denote all multiples of $ g $ as the set $S=\{s_j \mid s_j=g\cdot j, j \in \mathbb{Z^+}\}$. Now we define a subset of $ S $ satisfying some certain properties based on the memberships of no-instances: 
\[
\begin{array}{rcl}
	M & = & \{m + N \cdot i \mid  i,m \in \mathbb{Z^+}, \, m \bmod {p_j} \in \left[\frac{p_j}{8}, \frac{3p_j}{8}\right] \cup \left[\frac{5p_j}{8},\frac{7p_j}{8}\right] \mbox{ for } p_j\in R
	\\
	& & \mbox{ and } m \equiv 0 \,\bmod {p_j} \mbox{ for } p_j \in P_{k,n}\setminus R \}.
\end{array}
\]
The existence of such $m$'s follows from the Chinese remainder theorem. (Note that the set $M\subseteq S$ since the numbers $N$ and $m$ are multiple of $g$.)

Since $S'=S^{accept}$ and due to Markov chain theory, after certain threshold $ r_0 $, $\mathcal{P}$ accepts all input $a^{s_r}$ for any $s_r\in S $ and $ r\geq r_0$.  Moreover, for any $0 < \epsilon < \frac{1}{2}$, we can find an $ r_\epsilon $ such that
\[
 	| f_{\mathcal{P}} (a^{s_r}) -  f_{\mathcal{P}} (a^{s_{r'}}) |< \frac{1}{2} - \epsilon
\] for all $r, r' > r_\epsilon$.   

We pick $r,r'$ such that $ s_r $ is a multiple of $ N $ (and is in  $ S\setminus M$), and $ s_{r'} $ is in $ M $. It is clear that $a^{s_{r'}}\in \tt L^{k,n}_{no}$ and $a^{s_r}\in \tt L^{k,n}_{yes}$. Since $P_n$ is supposed to recognise $ \tt L^{k,n}$ with error bound $ \epsilon $, we must have 
 $ f_{\mathcal{P}} (a^{s_r}) \geq 1 - \epsilon$ and 
 $  f_{\mathcal{P}} (a^{s_{r'}}) \leq \epsilon$. We cannot get any value less than $ \frac{1}{2} $ if we check the maximum distance from $ f_{\mathcal{P}} (a^{s_r}) $ with radius $  \frac{1}{2} - \epsilon $. That means $ \mathcal{P} $ cannot solve $\tt L^{k,n}$ correctly.
This completes the proof of the claim.
$ \triangleleft $

Since $D$ is a multiple of $g$ we have that $D\geq g \geq p_n\cdot p_{n+1} \cdots p_{n+\left\lceil  k/3\right\rceil} $. 
Recall that  $D$ is the least common multiple of $t_l, \dots, t_l$, where $t_1, \dots, t_l$ are the lengths of cycles of Markov chain,  so it can be represented as $D=\prod_{s=1}^{r}p_s^{\max_{j=1}^l a_{j,s}},  $ where $t_j=p_1^{a_{j,1}}\cdot p_2^{a_{j,2}}\cdots p_r^{a_{j,r}}$ is the prime decomposition of number $t_j$.  From this we have that $t_1+\cdots +t_l\geq p_n + p_{n+1} + \cdots + p_{n+ \left\lceil \frac{k}{3}\right\rceil}.$

We can estimate the value of  $ p_n+\cdots +p_{n+ \left\lceil  \frac{k}{3} \right\rceil}$. It is known \cite{AT76} that  the $n$-th  prime number $p_n$ satisfies $p_n = \Theta( n\log(n))$. So we can follow that 
\small
\[
	\sum_{x=n}^{n+\left\lceil  \frac{k}{3} \right\rceil} p_x \geq c\int_{n-1}^{n+ \frac{k}{3}}x \log x\, dx\geq c' \frac{k}{3} \left(2n-1+\frac{k}{3} \right)\log\left(\frac{ n-1}{\sqrt{2}}\right)=\Omega(k(n+k)\log n).
\]
\normalsize

\section{The proof of Theorem \ref{thm:dfa-lowerbound}}
\label{app:dfa-lowerbound}
Our proof is similar to that of Theorem \ref{thm:state-bound-pfa} and we use also the idea given in \cite{AY12}. Let $T=p_n\cdot p_{n+1} \cdots p_{n+\left\lceil \frac{k}{3} \right\rceil}$. Assume that a DFA, say $ \mathcal{D} $, solves the promise problem $\tt L^{k,n}$ with less than $T$ states. Since $\tt L^{k,n}_{yes}$ and $\tt L^{k,n}_{no} $ contain infinitely many strings, there must exist a cycle of $t$ states ($t<T$) $s_0,\dots, s_{t-1}$ such that  $ \mathcal{D} $ visits this states in order
\[
	s_0 \rightarrow s_1 \rightarrow \cdots \rightarrow s_{t-1} \rightarrow s_0 \rightarrow \cdots 
\]

Without loss of generality suppose that $ \mathcal{D} $ enters the state $s_0$ after reading the yes-instance $a^{rN}$.  Let $S^{accept}=\{s_{i N \mod t}: i\geq 0\}$ be the set of states in which $ \mathcal{D} $ accepts the input. Let $d=\gcd(N,t)$ and let $t'=\frac{t}{d}.$ Since $N$ is a multiple of $d$ and $N=p_n\cdots p_{n+k-1}$, $d$ must be either 1 or the product of some $p$'s from the set $P_{k,n}$. 

Let $S'=\{s_{i d}:i=0,\dots,t'-1\}.$  We will show that $S'=S^{accept}$. Since $S^{accept}\subseteq S'$ and $|S'|=t'$ so we only need to show $|S^{accept}|\geq t'$ to obtain $S'=S^{accept}$.

Firstly we show that if $i$ satisfies $i N \mod t=0$, then $i$ must be a multiple of $t'$. If $i N \mod t \equiv 0$, then $i N= j t$ for some $ j \geq 0 $. Since $t=t'd$, $j = \frac{i N}{t' d}$. We know that $\gcd(N,t')=1$. Hence $i$ is a multiple $t'.$ 

For different $i_1,i_2 $ , i.e. $ 0\leq i_1<i_2 <t'$, we must have $i_1 N\neq i_2 N \pmod t$. Otherwise we can have $(i_1-i_2)N \equiv 0\pmod t$ and so $i_1-i_2 $ must be a multiple of $t'$. But, this is a contradiction.

Since we have a different value of  $(i N\mod t)$ for each $i \in \{0 \dots, t'-1\}$, $|S^{accept}|\geq t'$ and so  $S'=S^{accept}$.

Since $t<T=p_n\cdots p_{n+\left\lceil  \frac{k}{3}\right\rceil }$, $t$ can be divisible by less than  $\left\lceil \frac{k}{3}\right\rceil$ primes from the set $P_{k,n}$ and the same is true also for number $d.$ So there are at least $\frac{2k}{3}$ primes from $P_{k,n}$ which are not  multiples of $d$.  Let the set $R\subseteq P_{k,n}$ contains the primes not dividing $ d $, i.e. 
\[
	R=\{p_{i_j} \mid p_{i_j}\in P_{k,n}, d \not\equiv 0 \pmod {p_{i_j}} , j=1\dots,s\}~~\left(s\geq\frac{2k}{3}\right).
\]
 We denote all multiples of $ d $ as the set $S=\{s_j \mid s_j=d\cdot j, j \in \mathbb{Z^+}\}$. Now we define a subset of $ S $ satisfying some certain properties based on the memberships of no-instances: 
\[
\begin{array}{rcl}
	M & = & \{m + N \cdot i \mid  i,m \in \mathbb{Z^+}, m \bmod {p_j} \in \left[\frac{p_j}{8}, \frac{3p_j}{8}\right] \cup \left[\frac{5p_j}{8},\frac{7p_j}{8}\right] \mbox{ for } p_j\in R
	\\
	& & \mbox{ and } m \equiv 0 \,\bmod {p_j} \mbox{ for } p_j \in P_{k,n}\setminus R \}.
\end{array}
\]
The existence of such $m$'s follows from the Chinese remainder theorem. (Note that the set $M\subseteq S$ since the numbers $N$ and $m$ are multiple of $d$.)

The automaton $\mathcal{D}$ accepts all inputs $a^{s}$, where $s \in S$. Since $M\subset S$, $\mathcal{D}$ must accept the strings $a^j$, where $j \in M$. However, $a^j \in \tt L^{k,n}_{no}$. This is a contradiction. Therefore, the length of cycle $ t $ (and so the number of states required by $\mathcal{D}$) cannot be less than $ T $. 

Now, we calculate the value of $ T $ which is equal to $ p_n\cdots p_{n+\left\lceil  k/3\right\rceil}$. It is known \cite{AT76} that  the $n$-th prime number $p_n$ satisfies $p_n = \Theta( n\log(n))$. Then,
\[
	T = \prod_{x=n}^{n+\left\lceil  k/3\right\rceil}p_x\geq c\prod_{x=n}^{n+ k/3}x\log(x)\geq c'(n\log(n))^{k/3}=\Omega(n\log n)^{k/3}.
\]

\section{The proof of Theorem \ref{thm:dfa-upperbound}}
\label{app:dfa-upperbound}
Let $t=p_n\cdot p_{n+1} \cdots p_{n+\left\lfloor \frac{k}{3} \right\rfloor+1}$. The DFA $\mathcal{D}_{k,n}$ solving  the promise problem   $\tt L^{k,n}\in F_n$  has $t$ states $q_0,\dots, q_{t-1},$ where $q_0$ is an initial and the only accepting state. Reading an input  $\mathcal{D}_{k,n}$ moves from the state $q_i$ to $q_{i+1 \pmod t}$. It is clear that after processing any input $a^{m\cdot N} \in\tt  L^{k,n}_{yes}$ ($m \geq 0$) the automaton will be in the state $q_0$ and accepts the input. Reading the input $a^m$ automaton moves to the state $q_0$ if and only if $m$ is multiple of $p_n, p_{n+1}, \dots, p_{n+\left\lfloor \frac{k}{3} \right\rfloor+1}$ that means the number of primes which do not divide $m$ is less then $\frac{2k}{3}$ and $a^m\not\in \tt L^{k,n}_{no}$. It means  $\mathcal{D}_{k,n}$ solves the promise problem  $\tt L^{k,n} $ correctly.

Now, we calculate the value of $ t $ which is equal to $ p_n\cdots p_{n+\left\lfloor k/3\right\rfloor+1}$. Using that  the $n$-th prime number $p_n$ satisfies $p_n = \Theta( n\log(n))$ we have
\[
	t=\prod_{x=n}^{n+\left\lfloor k/3\right\rfloor+1}p_x\leq c\prod_{x=n}^{n+ k/3+1}x\log(x)\leq c'((n+k/3)\log(n+k/3))^{k/3}=\]
	\[O((n+k/3)\log (n+k/3))^{k/3}.
\]

\end{document}